\documentclass[11pt,letterpaper,english,preprintnumbers,amsmath,amssymb,superscriptaddress,nofootinbib,prl]{revtex4}

\usepackage{lipsum}
\usepackage{mathpazo}
\usepackage[T1]{fontenc}
\usepackage[latin9]{inputenc}
\usepackage{color}
\usepackage{verbatim}
\usepackage{amsmath}
\usepackage{amsthm}
\usepackage{amssymb}
\usepackage{graphicx}
\usepackage{setspace}

\makeatletter

\pdfpageheight\paperheight
\pdfpagewidth\paperwidth

\providecommand{\tabularnewline}{\\}

\@ifundefined{textcolor}{}
{%
 \definecolor{BLACK}{gray}{0}
 \definecolor{WHITE}{gray}{1}
 \definecolor{RED}{rgb}{1,0,0}
 \definecolor{GREEN}{rgb}{0,1,0}
 \definecolor{BLUE}{rgb}{0,0,1}
 \definecolor{CYAN}{cmyk}{1,0,0,0}
 \definecolor{MAGENTA}{cmyk}{0,1,0,0}
 \definecolor{YELLOW}{cmyk}{0,0,1,0}
}
\theoremstyle{remark}
\ifx\thechapter\undefined
  \newtheorem{rem}{\protect\remarkname}
\else
  \newtheorem{rem}{\protect\remarkname}[chapter]
\fi
\theoremstyle{plain}
\newtheorem{lyxalgorithm}{\protect\algorithmname}
\theoremstyle{definition}
\ifx\thechapter\undefined
  \newtheorem{defn}{\protect\definitionname}
\else
  \newtheorem{defn}{\protect\definitionname}[chapter]
\fi
\theoremstyle{plain}
\ifx\thechapter\undefined
  \newtheorem{lem}{\protect\Lemmaname}
\else
  \newtheorem{lem}{\protect\Lemmaname}[chapter]
\fi
\theoremstyle{plain}
\ifx\thechapter\undefined
  \newtheorem{cor}{\protect\corollaryname}
\else
  \newtheorem{cor}{\protect\corollaryname}[chapter]
\fi
\theoremstyle{plain}
\ifx\thechapter\undefined
  \newtheorem{thm}{\protect\theoremname}
\else
  \newtheorem{thm}{\protect\theoremname}[chapter]
\fi
\theoremstyle{plain}
\ifx\thechapter\undefined
  \newtheorem{conjecture}{\protect\conjecturename}
\else
  \newtheorem{conjecture}{\protect\conjecturename}[chapter]
\fi

\makeatother

\usepackage{babel}
\providecommand{\algorithmname}{Algorithm}
\providecommand{\conjecturename}{Conjecture}
\providecommand{\corollaryname}{Corollary}
\providecommand{\definitionname}{Definition}
\providecommand{\Lemmaname}{Lemma}
\providecommand{\remarkname}{Remark}
\providecommand{\theoremname}{Theorem}

\begin{document}

\title{Efficient unitary paths and quantum computational supremacy: \\
A  proof of average-case hardness of Random Circuit Sampling}
\author{Ramis Movassagh}
\email{q.eigenman@gmail.com}
\selectlanguage{english}%
\affiliation{IBM Research, MIT-IBM AI lab, Cambridge MA, 02142}
\maketitle
\begin{singlespace}
\noindent 
One-parameter interpolations between any two
unitary matrices (e.g., quantum gates) $U_1$ and $U_2$ along efficient paths contained in the
unitary group are constructed. Motivated by applications, we propose the continuous unitary
path $U(\theta)$ obtained from the QR-factorization
\[
U(\theta)R(\theta)=(1-\theta)A+\theta B,
\]
where $U_1 R_1=A$
and $U_2 R_2=B$ are the QR-factorizations of $A$ and $B$, and $U(\theta)$ is a unitary for all $\theta$ with $U(0)=U_1$ and $U(1)=U_2$. The
 QR-algorithm is modified to, instead of $U(\theta)$, output
a matrix whose columns are proportional to the corresponding columns
of $U(\theta)$ and whose entries are polynomial or rational functions of $\theta$.
By an extension of the Berlekamp-Welch algorithm we show that rational
functions can be efficiently and exactly interpolated with respect
to $\theta$. We then construct probability distributions over unitaries that are arbitrarily close to the Haar measure.

Demonstration of computational advantages of NISQ \cite{preskill2018quantum} over classical computers is an imperative near-term goal, especially with the exuberant experimental frontier in academia and industry (e.g., IBM and Google). A candidate for quantum computational supremacy is Random Circuit Sampling (RCS), which is the task of sampling from the output distribution of a random circuit. The aforementioned mathematical results provide a new way of scrambling quantum circuits and are applied to prove that exact RCS is $\#P$-Hard on average, which is a simpler alternative to Bouland {\it et al}'s \cite{bouland2018quantum}. (Dis)Proving the quantum supremacy conjecture requires \textit{approximate} average case hardness; this remains an open problem for all quantum supremacy proposals. 
\tableofcontents{}
\section{\label{sec:Unitary-paths-and}Unitary paths and summary of results}
\subsection{Paths on the unitary group}
Let $\mathbb{U}(N)$ denote the unitary group and suppose $U_{1}\in\mathbb{U}(N)$
and $U_{2}\in\mathbb{U}(N)$ are two unitary matrices. How can one
build a parametrized path $U(\theta)$ between them such that $U(\theta)\in\mathbb{U}(N)$
for all $\theta\in[0,1]$ and $U(0)=U_{1}$ and $U(1)=U_{2}$? In
this subsection we discuss a few different ways of constructing such
a $U(\theta)$.

Let the Hermitian matrix $H$ be $H=-i\log(U_{1}^{\dagger}U_{2})$,
then the most natural path is \textit{the geodesic}:
\begin{equation}
U(\theta)=U_{1}\exp(iH\theta),\label{eq:Geodesic}
\end{equation}
which connects $U_{1}=U(0)$ to $U_{2}=U_{1}e^{iH}=U(1)$. 

An alternative interpolation between the two unitaries is
\[
U(\theta)=U_{1}^{(1-\theta)}U_{2}^{\theta},\qquad\theta\in[0,1]
\]
where the power is defined, as for any normal matrix, to act only
on the eigenvalues. 

Motivated by the applications we have in mind we now propose a completely
different interpolation scheme that is inspired by numerical linear
algebra. In general any $N\times N$ matrix $M$ has a QR-decomposition
$UR=M$, where $U$ is a unitary (or orthogonal if $M$ is real) and
$R$ is an upper triangular matrix \cite{strang1993introduction}.
Taking the convention that the diagonal entries of $R$ are positive,
this decomposition is unique so long that the matrix is full rank.
Below we refer to both orthogonal and unitary simply as unitary.

Suppose one wants to interpolate between the unitaries in the QR decompositions
of the matrices $A=U_{1}R_{1}$ and $B=U_{2}R_{2}$ on a continuous
unitary path. Then, $U(\theta)$ defined by the QR-decomposition of the interpolation
is such a paths:
\begin{equation}
U(\theta)R(\theta)=(1-\theta)\text{ }A+\theta\text{ }B.\label{eq:MatrixInterpolation}
\end{equation}
By construction $U(\theta)$ is unitary for all $\theta$ (See Fig.
\ref{fig:Unitary_Interpolation} (left)).

Alternatively, the Euclidean line interval that connects any two unitary
matrices $U_{1}$ and $U_{2}$ is the matrix pencil $(1-\theta)U_{1}+\theta U_{2}$
for $\theta\in[0,1]$. Although this sum is not generally a unitary
matrix for $\theta\in(0,1)$, the interpolation
\begin{figure}
\begin{centering}
\includegraphics[scale=0.35]{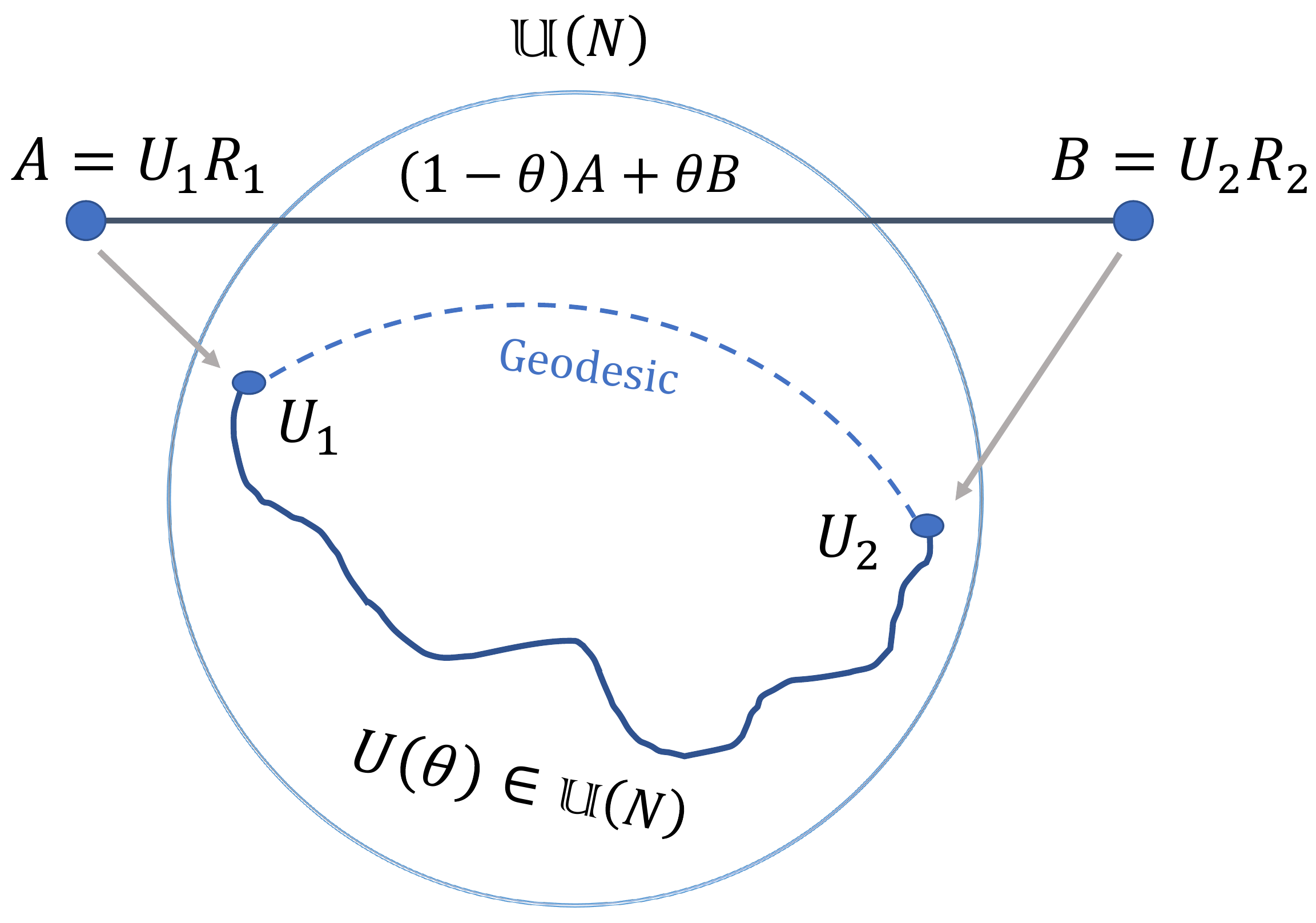}$\quad$\includegraphics[scale=0.35]{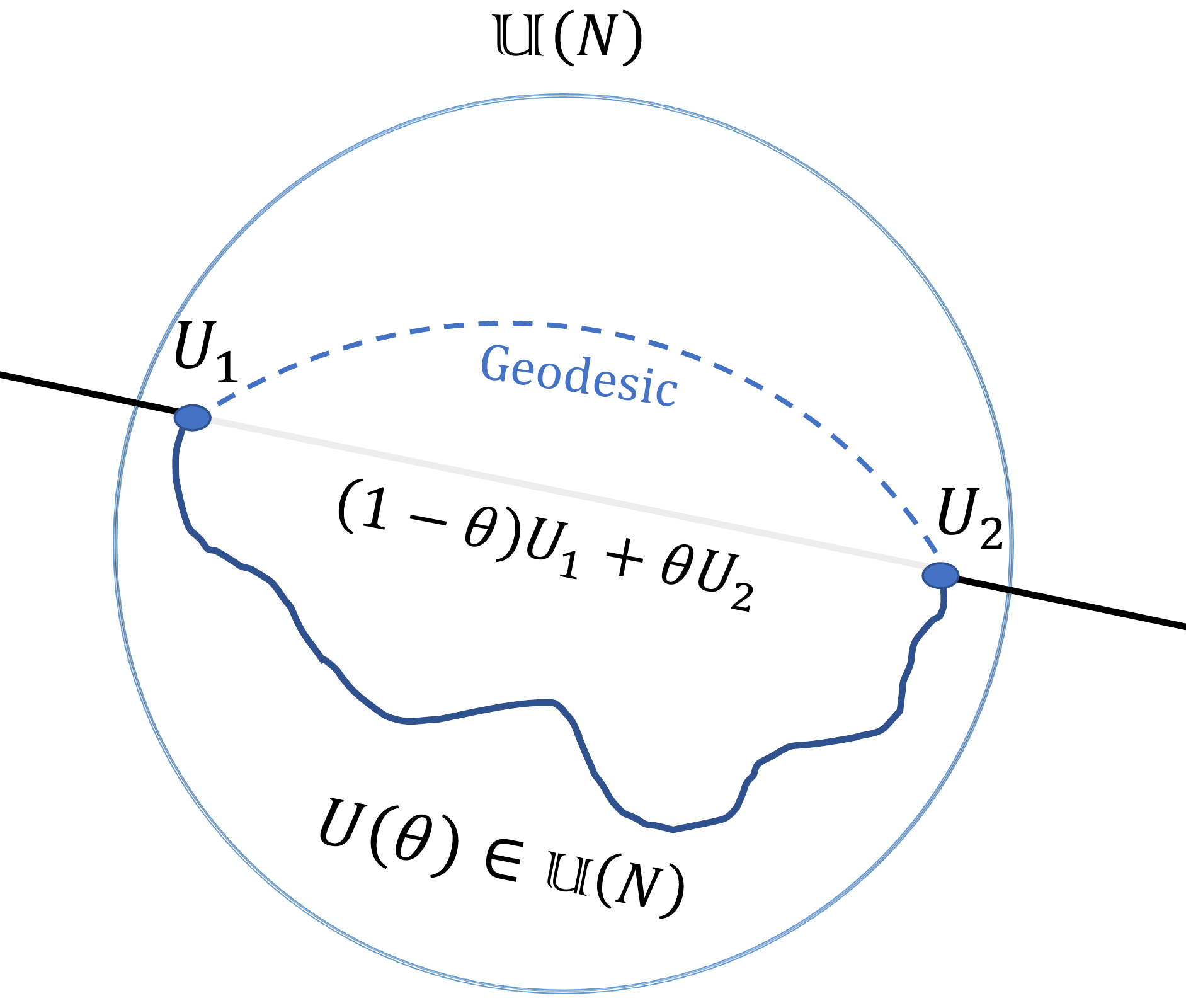}\caption{\label{fig:Unitary_Interpolation}Schematics of $U(\theta)\in\mathbb{U}(N)$ based on the interpolation given by: Left: Eq. \eqref{eq:MatrixInterpolation}
and Right: Eq. \eqref{eq:QR} }
\par\end{centering}
\end{figure}
\begin{equation}
U(\theta)R(\theta)=(1-\theta)\text{ }U_{1}+\theta\text{ }U_{2},\label{eq:QR}
\end{equation}
provides a unitary path $U(\theta)$ with the end points $U(0)=U_{1}$
and $U(1)=U_{2}$  (see Fig. \eqref{fig:Unitary_Interpolation} (right)).

Since for generic choice of matrices, the sums in Eqs. \eqref{eq:MatrixInterpolation} and \eqref{eq:QR}
are with probability one invertible, $U(\theta)$ is generically unique for all $\theta$.

Lastly, an interpolation is achieved by $U_{1}U(\theta)$ and the
QR-decomposition
\begin{equation}
U(\theta)R(\theta)=(1-\theta)\text{ }\mathbb{I}+\theta\text{ }U_{1}^{\dagger}U_{2}.\label{eq:QR_line}
\end{equation}
This decomposition is guaranteed to be unique for all $\theta\ne1/2$
and satisfies $U(0)=U_{1}$ and $U(1)=U_{2}$. Eq. \eqref{eq:QR_line}
``multiplicatively'' interpolates between the unitaries, which is
an alternative to the exponential form (Eq. \eqref{eq:Geodesic}).

Why do we bother with such proposals when we have Eq. \eqref{eq:Geodesic}?
In applications, such as in (quantum) complexity theory (Section \eqref{sec:RCS}),
one may access the unitary matrices $U(\theta_{1})$, $U(\theta_{2}),\dots,U(\theta_{D})$
for some $D$, and want to determine $U(\theta)$ for all
$\theta$. That is, the path $U(\theta)$ needs to be \textit{learned}
by sampling at some number of points.
$U(\theta)$ can be efficiently determined from samples if, for example,
the entries of $U(\theta)$ have a polynomial structure of degree
$D-1$. More generally, one wishes for $U(\theta)$ to have entries
whose functional dependence on $\theta$ are such that they can be
inferred efficiently by sampling even if the sample points are noisy. 
\subsection{\label{sec:Summary_overview}Motivation and summary of the results}
The rest of this paper has two main sections: \textbf{Section \eqref{sec:Unitary_Homotopy}}
is self-contained and may be read by anyone interested in the structure
of the unitary group. Starting in \textbf{Section \eqref{sec:RCS}} we apply the mathematical results to a problem in quantum complexity theory.\\

\textbf{Section \eqref{sec:Unitary_Homotopy}} includes an explicit
and efficient construction of one-parameter families of unitary matrices
that interpolate between any two fixed unitaries. By 'efficient' we
mean that a modified QR-algorithm is presented (Alg. \eqref{alg:ModifiedQR}) that results
in unitaries whose entries are proportional to polynomial functions of low degrees in the
interpolation parameter. The degrees
of the polynomial functions are quantified (Lemma \eqref{Lem:RationalFunction} and Corollary \eqref{cor:Polynomial}). 
Therefore, we have
in our hands an interpolation scheme between any two given unitaries
that is continuous, is contained in the unitary group everywhere, and has entries that depend on polynomials of low degrees.

The discovery of Berlekamp-Welch algorithm (BW) \cite{welch1986error}
was originally motivated by classical error correction schemes such
as the Reed-Solomon codes \cite{reed1960polynomial}. In these codes,
the messages are encoded in the coefficients of polynomials over finite
fields. BW takes as input a set of points together with the value
of the polynomial at those points. The latter depends on the encoded
messages, which may incur errors in the transmission from the sender
to the receiver. BW is remarkable in that it can exactly recover such
polynomials even if the evaluation of the polynomial at some number
of points is erroneous. Here we prove an extension of BW that can
efficiently interpolate rational functions (Alg. \eqref{alg:(Berlekamp-Welch-for-Rational}).

Lastly, a subsection is devoted to the Haar measure, and the construction
of unitaries whose distributions may be arbitrarily close to the Haar
measure in total variational distance (TVD). We discuss the intimate connection of the Haar measure
 to the QR-factorization of random gaussian
matrices and then provide a new construction of distributions over unitary matrices
that are $\theta-$close in TVD to the Haar measure for any given
small $\theta$. \\

In \textbf{Section \eqref{sec:RCS}} we apply the above mathematical
results to prove a result in quantum complexity theory. Before stating
the results, we provide an introduction to the field.

A quantum computation is a rotation in the Hilbert space of the standard
basis $e_{1}=(1,0,\dots,0)^{T}$. The standard notation denotes $e_{1}$
by $|0\rangle^{\otimes n}$, which is a $2^{n}$ dimensional vector;
the rotation is achieved by the unitary matrix $U$, whose exact form
is fixed by the quantum algorithm being implemented. The final state
of the quantum computation, denoted by $|\psi\rangle$, is the vector
in the Hilbert space $|\psi\rangle\equiv U\left(|0\rangle^{\otimes n}\right)$.
The square of the absolute value of any entry of $|\psi\rangle$ quantifies
the probability of occurrence of that particular outcome upon measurement.
Good quantum algorithms provide $U$'s that, with high probability,
result in outputs that encode the answer to the desired computational
task more efficiently than any classical computer.

One then says that the quantum computation involves $n$ quantum bits
(qubits), each of which corresponds to $\mathbb{C}^{2}$, and that
the circuit $C$ implements $U$. Constrained by the difficulties in the
experimental realizations, $U$ is almost always taken to be a product
of many 'local' unitaries (gates) of the form $\mathbb{I}_{2^{n-1}}\otimes U_{i}$
or $\mathbb{I}_{2^{n-2}}\otimes U_{i,j}$ where $U_{i}$ and $U_{i,j}$
are $2\times2$ and $4\times4$ unitaries respectively. $U_{i}$ acts
on the Hilbert space corresponding to the $i^{\text{th}}$ qubit and
$U_{i,j}$ acts on the joint Hilbert space of the neighboring qubits
$i$ and $j$. One says that the circuit $U$ is \textit{generic}
with respect to the architecture ${\cal A}$ if the local unitaries
$U_{i}$ and $U_{i,j}$ are drawn independently from the Haar measure. 

Currently, one of the most active frontiers in quantum computation
and information science is centered around ``Noisy Intermediate Scale
Quantum (NISQ)'' computers \cite{preskill2018quantum} -{}-noisy
because of the absence of quantum error correction. What can a NISQ
computer do? Would these offer a quantum advantage over any current
classical computer? 

Because of a large industrial push (e.g., from IBM and Google), NISQ
computers with hundred(s) of qubits are at the brink of existence
with the promise of outperforming any classical computer \cite{dalzell2018many}.
A milestone is to prove unambiguously a substantial advantage of a
NISQ computer over classical ones. This event has been termed quantum
supremacy, and we have yet to witness it. It states that there are
computational tasks that a NISQ computer can efficiently perform,
that any classical computer would find formidable. The two main current proposals
to demonstrate this are \texttt{BosonSampling} \cite{aaronson2011computational},
and Random Circuit Sampling (RCS). Our focus in this paper will be
RCS, which roughly speaking, is the task of sampling from the output
distribution of a circuit whose local gates are random. 

Recently \textit{generic} aspects of quantum circuits have been foundational.
Generic in mathematics means with probability one or almost surely.
Generic circuits are therefore circuits whose local gates are random
unitaries. Most natural is to draw them from the Haar measure-- drawn
independently and uniformly from the space of all unitaries. Such
circuits appear in quantum complexity theory \cite{aaronson2016complexity},
the study of black holes \cite{hayden2007black}, holographic models
of quantum gravity \cite{takayanagi2018holographic}, and quantum
supremacy \cite{harrow2017quantum,dalzell2018many}. The latter was
initially proposed to demonstrate that quantum computers have capabilities
beyond classical. Aaronson and Arkhipov first proposed
the \texttt{BosonSampling} problem as a candidate for testing quantum
supremacy \cite{aaronson2011computational}. Later the Google team proposed that random quantum circuits
might demonstrated this supremacy in the near-term quantum devices
\cite{boixo2018characterizing}. We currently know that computing the output probabilities of RCS, even
approximately, is $\#P$-Hard\footnote{\#P is a generalization of NP, which is for decision problems, to
counting problems. In particular, an NP-complete problem is: Does
a 3-SAT instance have any solutions? The answer is a yes or no (i.e.,
zero solutions). Whereas, \#P asks: How many solutions does a 3-SAT
instance have? Therefore, \#P contains NP. } in the worst case \cite{bremner2011classical}. This implies that no exact worst-case classical simulation algorithm exists unless the polynomial hierarchy collapses \cite{aaronson2011computational, bremner2011classical}. An advantage of RCS over \texttt{BosonSampling} is that anti-concentration bounds have been proven \cite{Harrow_Saeed2018,Brandao2013}. This led Bouland et al to prove that the exact average-case
RCS is also $\#P$-Hard \cite{bouland2018quantum} . What is needed and remains open in all supremacy
proposals to date, however, is to prove \textit{approximate} average-case
hardness. The status of the field is summarized in the table below.
\begin{center}
\begin{tabular}{|c|c|c|}
\hline 
 & Exact & Approximate\tabularnewline
\hline 
\hline 
Worst case & \#P & \#P\tabularnewline
\hline 
Average case & \#P & ?\tabularnewline
\hline 
\end{tabular}
\par\end{center}
The application that motivates this paper is to (dis)prove the so
called quantum supremacy conjecture, which we pick up in \textbf{Section
\ref{sec:RCS}}. Although the conjecture remains open, we provide
an entirely a new proof of the exact average case hardness of RCS,
which is alternative to the recent paper by Bouland et al \cite{bouland2018quantum}
and does not use the standard techniques. The standard technique for
proving average case hardness is via polynomial reduction of the worst case hardness to the average, which is based on Lipton's technique \cite{lipton1989new}. This
 essentially says:

\textit{Suppose a problem is known to be $\#P$-Hard in the worst case.
Now suppose the worst case problem, $A$, can be deformed to $A(\theta)$
for $\theta\in[0,1]$ such that $A(1)=A$ is the $\#P$-Hard instance,
and $A(\theta)$ is a polynomial of low degree $D$ in $\theta$.
If $\theta\in[0,1)$ contains at least $D+1$ 'generic' instances
of the problem, then the problem is generically $\#P$-Hard. The reason
is that if the problem were generically easy, then because of the polynomial
structure, the evaluation of $A(\theta)$ at $D+1$ random points
would uniquely determine $A(\theta)$. One can then read off $A(1)$
by polynomial extrapolation. But $A(1)$ is a $\#P$-Hard problem so
the generic instances must have had $\#P$-Hard instances.}

To make use of this polynomial interpolation technique, Aaronson and Arkhipov proposed the BosonSampling, whose underlying computational task is a pemanant, which is a polynomial function of the entries \cite{aaronson2011computational}. 

In RCS, the goal to demonstrate quantum supremacy is to compute 
\begin{equation}
\text{p}_{y}(C)\equiv|\langle y|C|0^{n}\rangle|^{2}.\label{eq:py}
\end{equation}

Quantum circuits, however, do not have a natural polynomial  form. To obtain polynomials corresponding to the generic local circuit, Bouland et al \cite{bouland2018quantum}
deform the quantum gates towards a Haar distribution and then truncate the Taylor series expansions of the exponential functions that arise in Eq. \eqref{eq:Geodesic}.
The truncation results in 'non-unitarity' of the local gates and therefore
the quantum circuit as a whole. They then justify that the errors
due to the truncations are small enough that the average case $\#P$-
Hardness of the non-unitary approximation is still necessary for the approximate average-case hardness conjecture \eqref{Conj1} to be true.

In Section \eqref{sec:RCS}, the results of Section \eqref{sec:Unitary_Homotopy}
are used to give an entirely a new proof of $\#P$-hardness of RCS. We prove that the modified QR algorithm results in probabilities (Eq. \eqref{eq:py}) that are rational functions of the interpolation parameter $\theta$. These can then be efficiently learned using the new BW algorithm for rational functions. Moreover, we prove that proving $\# P$-hardness of  exact computation of Eq. \eqref{eq:py} is necessary for proving the quantum supremacy conjecture. This is ensured by our construction as the local gates will be proved to be $\theta-$close to the Haar distribution in TVD.

In addition to being a more direct proof, the advantages of the new proof over \cite{bouland2018quantum}
include: \\
${ }(1)$ No truncations are needed and the interpolation
is contained in the unitary group everywhere.\\
${ }(2)$ The construction is explicit in that the degrees and coefficients can be quantified;
this might help in proving the quantum supremacy conjecture (see
 \cite[(arXiv version, page 81)]{aaronson2011computational}).\\
${ }(3)$ The new mathematical results may be of independent interest. 
\section{\label{sec:Unitary_Homotopy}Rational and polynomial paths on the unitary
group}
\subsection{\label{sec:Standard_QR}Standard QR decomposition}
The standard QR-decomposition algorithm applied to the columns of
a matrix, $M$, results in an orthonormal set of vectors that comprise
the columns of a unitary matrix $U$, whose linear combinations dictated
by the $R$ matrix result in $M=UR$. It is instructive to quickly
recap this algorithm.

Let $A=[a_{1},\dots,a_{N}]$, $B=[b_{1},\dots,b_{N}]$ where $a_{i}$
and $b_{i}$ are the $i^{th}$ columns of $A$ and $B$ respectively.
Let $M(\theta)=(1-\theta)A+\theta B$ have the QR-decomposition $M(\theta)=U(\theta)R(\theta)$.
Let $m_{i}(\theta)$ be the columns of $M(\theta)=[m_{1}(\theta),\dots,m_{N}(\theta)]$
and denote the columns of $U(\theta)$ by $U(\theta)=[u_{1}(\theta),\dots,u_{N}(\theta)]$.

Comment: Below at times we drop the dependence on $\theta$ for simplicity
and denote the column vectors $m(\theta)$ simply by $m$ etc. It
is obvious that whenever the matrix depends on $\theta$, so do its
columns.

Also recall that the projection of the vector $x$ onto $y$ is
\[
\text{proj}_{y}x=\frac{\langle y,x\rangle}{\left\Vert y\right\Vert ^{2}}y,
\]
where $\left\Vert \centerdot\right\Vert $ denotes the standard Euclidean
$2-$norm. 

The algorithm starts by $v_{1}\equiv m_{1}=(1-\theta)a_{1}+\theta b_{1}$.
Let $u_{1}=v_{1}/\left\Vert v_{1}\right\Vert $, then $v_{2}=m_{2}-\text{proj}_{u_{1}}m_{2}$
and $u_{2}=v_{2}/\left\Vert v_{2}\right\Vert $. Generally,
\[
v_{k}=m_{k}-\sum_{j=1}^{k-1}\text{proj}_{uj}m_{k},\qquad u_{k}=v_{k}/\left\Vert v_{k}\right\Vert .
\]
It is easy to see that the functional form of the entries of any $u_{i}$
is a ratio, whose numerator and denominator are in general sums of
terms some of which involve square roots of polynomials. For example
\begin{align*}
u_{1} & =\frac{(1-\theta)a_{1}+\theta b_{1}}{\sqrt{(1-\theta)^{2}\left\Vert a_{1}\right\Vert ^{2}+\theta^{2}\left\Vert b_{1}\right\Vert ^{2}+\theta(1-\theta)\left[\langle a_{1},b_{1}\rangle+\langle b_{1},a_{1}\rangle\right]}};\\
v_{2} & =\left[(1-\theta)a_{2}+\theta b_{2}\right]-\frac{\theta(1-\theta)\left(\langle a_{1},b_{2}\rangle+\langle b_{2},a_{1}\rangle\right)\left[(1-\theta)a_{1}+\theta b_{1}\right]}{(1-\theta)^{2}\left\Vert a_{1}\right\Vert ^{2}+\theta^{2}\left\Vert b_{1}\right\Vert ^{2}+\theta(1-\theta)\left[\langle a_{1},b_{1}\rangle+\langle b_{1},a_{1}\rangle\right]},
\end{align*}
and $u_{2}=v_{2}/\left\Vert v_{2}\right\Vert $ would involve square
roots of polynomials. The standard QR-algorithm, therefore, is inadequate
for the applications we have in mind as the entries involve untamable
square roots of polynomials that cannot be learned by sampling.
\begin{rem}
In numerical linear algebra the QR-decomposition is performed by more
preferred methods that are more efficient and stable such as
Householder transformations, or Givens rotations \cite{trefethen1997numerical}.
Although these methods, in the standard form, do not yield any polynomial
structure in $U(\theta)$ to be exploited, they can be used to evaluate
the numerical values of $U(\theta)$ for any fixed $\theta$.
\end{rem}
Suppose we have a matrix $A$ that is $\ell\times k$ with $\ell\ge k$.
The computational complexity of QR decomposition depends on the algorithm
and is \cite{trefethen1997numerical}:
\begin{itemize}
\item Householder transformations: $2\ell k^{2}-2k^{3}/3$.
\item Standard QR decomposition procedure: $2\ell k^{2}$.
\item Modified QR decomposition (below) after normalization: $2\ell k^{2}$.
\end{itemize}
\subsection{\label{sec:SolvingQR}Modified QR decomposition and rational function
Berlekamp-Welch interpolation algorithm}
We propose an unnormalized version of the QR-decomposition
algorithm, which when applied to pencils of matrices in Eq. \eqref{eq:QR},
results in a matrix whose columns are orthogonal but not normalized.
The upshot is that the entries of this matrix can be expressed as polynomial or rational functions
of $\theta$.  Below we use the same notion for the columns of the matrices as we did in the standard QR-algorithm above.
\begin{lyxalgorithm}
\label{alg:ModifiedQR}Let $M(\theta)=(1-\theta)A+\theta B$, we seek
the QR-decomposition of $U(\theta)R(\theta)=M(\theta)$. 

The algorithm first solves for the unnormalized vectors $z_{i}$ by
performing the following:
\begin{enumerate}
\item Let $z_{1}=v_{1}=m_{1}$, which is linear in $\theta$.
\item $v_{2}=m_{2}-\text{proj}_{z_{1}}m_{2}$, it is instructive to write
it explicitly
\[
v_{2}=\frac{\left[(1-\theta)a_{2}+\theta b_{2}\right]\left\Vert z_{1}\right\Vert ^{2}-\langle z_{1},m_{2}\rangle z_{1}}{\left\Vert z_{1}\right\Vert ^{2}}
\]
where 
\begin{eqnarray*}
\left\Vert z_{1}\right\Vert ^{2} & = & (1-\theta)^{2}\left\Vert a_{1}\right\Vert ^{2}+\theta^{2}\left\Vert b_{1}\right\Vert ^{2}+\theta(1-\theta)\left[\langle a_{1},b_{1}\rangle+\langle b_{1},a_{1}\rangle\right]\\
\langle z_{1},m_{2}\rangle z_{1} & = & \left[(1-\theta)^{2}\langle a_{1},a_{2}\rangle+\theta^{2}\langle b_{1},b_{2}\rangle+\theta(1-\theta)\left(\langle a_{1},b_{2}\rangle+\langle b_{1},a_{2}\rangle\right)\right]\left[(1-\theta)a_{1}+\theta b_{1}\right]
\end{eqnarray*}
Therefore, every entry of $v_{2}$ is a rational function with numerator
being a polynomial of degree three and denominator a polynomial of
degree two. We can now define 
\[
z_{2}=\left\Vert z_{1}\right\Vert ^{2}v_{2}
\]
which is a polynomial valued vector, whose entries are polynomials
of degree three.
\item In general,\begin{eqnarray}
v_{k} & = & m_{k}-\sum_{j=1}^{k-1}\text{proj}_{z_{j}}m_{k}.\label{eq:v_k}\\
z_{k} & = & \left(\prod_{j=1}^{k-1}\left\Vert z_{j}\right\Vert ^{2}\right)\left\{ m_{k}-\sum_{j=1}^{k-1}\frac{\langle z_{j},m_{k}\rangle z_{j}}{\left\Vert z_{j}\right\Vert ^{2}}\right\} \label{eq:z_k}
\end{eqnarray}
The set of vectors $v_{1},\dots,v_{N}$ are orthogonal but \textit{unnormalized.}
We shall keep them unnormalized to retain the rational function dependence
of the entries on $\theta$ (see Lemma \ref{Lem:RationalFunction}
and Corollary \ref{cor:Polynomial}). Similarly, the set of vectors
$z_{1},\dots,z_{N}$ are orthogonal but unnormalized\textit{ }and
have entries that are polynomials in $\theta$. They are obtained by simply
multiplying through by the common denominator of all $v_{j}$'s.
For example, $z_{1}=m_{1}$, $z_{2}=v_{2}\left\Vert m_{1}\right\Vert ^{2}$
etc..
\item (Optional) Given $z_{1},\dots,z_{N}$, the columns of $U(\theta)$
are obtained by normalizing $z_{1},\dots,z_{N}$. Each entry of $U(\theta)$
will be a ratio of a polynomial function in $\theta$ with the square
root of a polynomial function in $\theta$:
\begin{equation}
u_{k}(\theta)=z_{k}(\theta)/\sqrt{\langle z_{k}(\theta),z_{k}(\theta)\rangle},\quad1\le k\le N.\label{eq:Entries_ModifiedQR}
\end{equation}
This is sufficient for our purposes as we will need the norm-square
of such quantities which are exactly rational functions.
\end{enumerate}
\end{lyxalgorithm}
\begin{defn}
Let $\mathbf{p}_{\ell}(\theta)$ be the set of polynomial \textit{vectors}
of degree $\ell$ in $\theta$; that is, every entry of the vector
$v\in\mathbf{p}_{\ell}(\theta)$ is a polynomial of degree $\ell$
in $\theta$. Let $q_{\ell}(\theta)$ be the set of polynomial \textit{functions}
of degree $\ell$ in $\theta$. The algebra is such that for $v\in\mathbf{p}_{\ell}(\theta)$,
$u\in\mathbf{p}_{r}(\theta)$, $f\in q_{\ell}(\theta)$, and $g\in q_{r}(\theta)$,
we have $fg\in q_{\ell+r}(\theta)$, $\langle u,v\rangle\in q_{\ell+r}(\theta)$,
and $g\text{ }v\in\mathbf{p}_{\ell+r}(\theta)$.
\end{defn}
\begin{defn}
A vector $v(\theta)$ whose entries are rational functions of $\theta$
is denoted by $v(\theta)\in\mathbf{p}_{\ell}(\theta)/q_{r}(\theta)$
if its entries have numerators that are polynomials of degree $\ell$
and denominators are polynomials of degree $r$. 
\end{defn}
\begin{lem}
\label{Lem:RationalFunction}$v_{k}(\theta)$ given by Eq. \eqref{eq:v_k}
in the modified QR algorithm satisfies $v_{k}(\theta)\in\mathbf{p}_{D_{k}+1}(\theta)/q_{D_{k}}(\theta)$,
with $D_{k}=(3^{k-1}-1)$. 
\end{lem}
\begin{proof}
We first prove that $v_{k}(\theta)\in\mathbf{p}_{D_{k}+1}(\theta)/q_{D_{k}}(\theta)$
for some positive integer $D_{k}$. Note that $v_{1}(\theta)=z_{1}(\theta)=m_{1}(\theta)\in\mathbf{p_{1}}(\theta)/q_{0}(\theta)$,
suppose $v_{k-1}(\theta)\in\mathbf{p}_{D_{k-1}+1}(\theta)/q_{D_{k-1}}(\theta)$,
then $z_{k-1}\in\mathbf{p}_{D_{k-1}+1}(\theta)$ and by Eq. \eqref{eq:v_k}
\[
v_{k}=\left\{ m_{k}-\sum_{j=1}^{k-2}\text{proj}_{z_{j}}m_{k}\right\} -\text{proj}_{z_{k-1}}m_{k}
\]
and the sum of the terms in the braces are of the same type as $v_{k-1}(\theta)$
by construction. Hence, by the induction hypothesis they are of the
type $\mathbf{p}_{D_{k-1}+1}(\theta)/q_{D_{k-1}}(\theta)$. Further,
$\text{proj}_{z_{k-1}}m_{k}=\langle z_{k-1},m_{k}\rangle z_{k-1}/\left\Vert z_{k-1}\right\Vert ^{2}\in\mathbf{p}_{2D_{k-1}+3}(\theta)/q_{2D_{k-1}+2}(\theta)$.
So we have
\[
v_{k}\in\frac{\mathbf{p}_{D_{k-1}+1}(\theta)}{q_{D_{k-1}}(\theta)}-\frac{\mathbf{p}_{2D_{k-1}+3}(\theta)}{q_{2D_{k-1}+2}(\theta)}\in\frac{\mathbf{p}_{3D_{k-1}+3}(\theta)}{q_{3D_{k-1}+2}(\theta)}.
\]
Letting $D_{k}=3D_{k-1}+2$ we arrive at the desired result. Lastly,
$D_{k}=(3^{k-1}-1)$ is the solution of the recursion relation $D_{k+1}=3D_{k}+2$
with the initial condition $D_{1}=0$.
\end{proof}
\begin{cor}
\label{cor:Polynomial}Let $z_{k}=q_{D_{k}}(\theta)\text{ }v_{k}$,
then $z_{k}\in\mathbf{p}_{D_{k}+1}(\theta)$. That is, each entry
of $z_{k}$ in Eq. \eqref{eq:z_k} is a polynomial function of degree $D_{k}+1=3^{k-1}$.
\end{cor}
\begin{rem}
Because the local gates in a quantum circuit that implements a quantum
computation are represented by $2\times2$ or $4\times4$ (unitary)
matrices, we note that $D_{2}+1=4$ and $D_{4}+1=28$.\\
\end{rem} 
We now turn to the issue of uniquely determining a rational function
by efficient sampling.
\begin{lem}
\label{fact:Rational-function}Any rational function of degree $(k_{1},k_{2})$
in one variable $\theta$ has the general form
\[
F(\theta)=\frac{a_{k_{1}}\theta^{k_{1}}+a_{k_{1}-1}\theta^{k_{1}-1}+\cdots+a_{0}}{b_{k_{2}}\theta^{k_{2}}+b_{k_{2}-1}\theta^{k_{2}-1}+\cdots+b_{0}}
\]
and is uniquely determined by $k_{1}+k_{2}+1$ points provided that
$F(\theta_{i})=f_{i}<\infty$ for $i\in[k_{1}+k_{2}+1]$ are independent
conditions and the numerator and denominator are relatively prime.\end{lem}
\begin{proof}
Since a rational function is determined up to a constant multiple
of numerator and denominator, we can factor out $a_{k_{1}}/b_{k_{2}}$,
and use homogeneous coordinates with $k_{1}+k_{2}+1$ unknowns. By
multiplying both sides by the denominator and then evaluating $F(\theta)$
at $k_{1}+k_{2}+1$ points $F(\theta_{i})=f_{i}$, the coefficients
become the solution of the linear system of equations in $(k_{1}+k_{2}+1)$
variables. Given that the $f_{i}$ are independent, the coefficients
are uniquely determined (unique point of intersection of hyperplanes).
Lastly, $f_{i}<\infty$ is to emphasize that we discard any $\theta_{i}$
that is a root of the denominator; such $\theta$'s are of measure
zero anyway.
\end{proof}
In coding theory, and especially in Reed-Solomon codes \cite{reed1960polynomial},
the messages $a_{1},\dots,a_{k}$ may be encoded into a polynomial
$a_{1}+a_{2}\theta+\cdots+a_{k}\theta^{k}$, which then is evaluated
at $n>k+t$ points. Then, the decoding procedure recovers the polynomial
and hence the message exactly despite $t$ errors. The decoding procedure
relies on Berlekamp-Welch (BW) algorithm for polynomial interpolation
\cite{welch1986error,SudanLect}. BW can be extended to interpolate
rational functions. The proof follows Sudan's and is a generalization
of it from polynomial to rational functions \cite{SudanLect}.
\begin{defn}
(Error polynomial) \label{Def: ErrorPolynomial}Suppose $f=(f_{1},\dots,f_{n})$
is a vector. Let $F(\theta)$ be a rational function of degree $(k_{1},k_{2})$.
We define the error polynomial $E(\theta)$ as one that satisfies
\[
E(\theta_{i})=0\quad\text{if}\quad F(\theta_{i})\ne f_{i},\quad\text{deg}(E(\theta))\le t.
\]
\end{defn}
\begin{lyxalgorithm}
\label{alg:(Berlekamp-Welch-for-Rational}(Berlekamp-Welch for Rational
Functions) Given $(\theta_{1},f_{1})$, $(\theta_{2},f_{2})$, ...,
$(\theta_{n},f_{n})$, find a rational function $F(\theta)$ of degree
$(k_{1},k_{2})$ exactly by evaluating it at $n>k_{1}+k_{2}+2t$ points
despite $t$ errors in the evaluation points:
\[
|\{i\in[n]|F(\theta_{i})\ne f_{i}\}|\le t.
\]
\end{lyxalgorithm}
\begin{proof}
The error polynomial by Def. \eqref{Def: ErrorPolynomial} satisfies
\begin{equation}
E(\theta_{i})F(\theta_{i})=E(\theta_{i})f_{i}.\label{eq:ErrorPoly}
\end{equation}
Let $W(\theta_{i})\equiv E(\theta_{i})f_{i}$, which implies that
$f_{i}=W(\theta_{i})/E(\theta_{i})$. Since $W(\theta)=E(\theta)F(\theta)$
is a $(k_{1}+t,k_{2})$ rational function, by Eq. \eqref{eq:ErrorPoly},
$f_{i}$ is a $(k_{1}+t,k_{2}+t)$ rational function of $\theta$.
By Lemma \eqref{fact:Rational-function}, the linear system defined
by Eq. \eqref{eq:ErrorPoly}, has a solution as long as $n>k_{1}+k_{2}+2t$.
If $W(\theta)/E(\theta)$ results in a rational function of degree
$(k_{1},k_{2})$ we are done and we simply output it as $F(\theta)$,
otherwise we decide that there were too many errors. 

Can the algorithm find distinct $(E_{1}(\theta),W_{1}(\theta))$ and
$(E_{2}(\theta),W_{2}(\theta))$? We now show that $W_{1}(\theta)/E_{1}(\theta)$
and $W_{2}(\theta)/E_{2}(\theta)$ are equal, which means that $F(\theta)$
is learned uniquely even if there are multiple solutions. We have
\[
\frac{W_{1}(\theta)}{E_{1}(\theta)}=\frac{W_{2}(\theta)}{E_{2}(\theta)}\iff E_{1}(\theta)W_{2}(\theta)=E_{2}(\theta)W_{1}(\theta).
\]
Recall that both sides are bounded degree rational functions (of degree
$(k_{1}+2t,k_{2})$). So by evaluating them at enough points we can
determine them uniquely (Lemma \eqref{fact:Rational-function}). Since
at every $\theta_{i}$
\[
E_{1}(\theta_{i})f_{i}=W_{1}(\theta_{i}),\quad E_{2}(\theta_{i})f_{i}=W_{2}(\theta_{i}),
\]
solving for $f_{i}$, we have $E_{1}(\theta_{i})W_{2}(\theta_{i})f_{i}=f_{i}E_{2}(\theta_{i})W_{1}(\theta_{i})$
at every $\theta_{i}$. If $f_{i}=0$, then by Eq. \eqref{eq:ErrorPoly}
$W(\theta_{i})=0$, otherwise we just cancel the $f_{i}$. This proves
the claim that $F(\theta)=E_{1}(\theta)/W_{1}(\theta)=E_{2}(\theta)/W_{2}(\theta)$.
Since $n>k_{1}+k_{2}+2t$, we are guaranteed to have enough points.
 \end{proof}
\begin{rem}
 Here we gave an algorithm in which the number of errors $t<(n-k_{1}-k_{2})/2$,
but as in standard BW algorithm, it is entirely possible to find algorithms
that can handle more errors. The above is sufficient for our purposes
so we leave finding such algorithms for future work.
 \end{rem}
 \begin{rem}
 As in the standard BW algorithm, this algorithm works just as well
over finite fields, reals or complex numbers. As far as we know, the
performance with respect to approximate $f_{i}$'s is not
known.
\end{rem}
\subsection{\label{sec:Close_to_Haar} The Haar measure and distributions close
to it in total variational distance (TVD)}
 Let $\mathbb{O}(N)$, $\mathbb{U}(N)$, and $\mathbb{SP}(2N)$ denote
the set of orthogonal, unitary, and symplectic matrices respectively.
The entries of these matrices are drawn from real ($\beta=1$), complex
($\beta=2$), and ($\beta=4$) respectively. In what follows we ignore
$\beta=4$ case. The set of such matrices with determinant equal to
one are, respectively, denoted by $\mathbb{SO}(N)$ and $\mathbb{SU}(N)$.
Moreover, $\mathbb{O}(N)$, $\mathbb{SO}(N)$ form subgroups of the
set of $N\times N$ real matrices. Similarly, $\mathbb{U}(N)$ and
$\mathbb{SU}(N)$ form subgroups of the set of $N\times N$ complex
matrices. If $G$ is any one of the matrix groups defined above, then
a \textit{uniform random element} of $G$ is a matrix $V\in G$ whose
distribution is\textit{ translation invariant}, which means for any
fixed $M\in G$,
\[
VM\stackrel{d}{=}MV\stackrel{d}{=}V,
\]
where $\stackrel{d}{=}$ denotes equality in the distribution sense.
Alternatively, the the distribution of a uniform random element of
$G$ is a translation invariant probability measure. Below we denote
both ``orthogonal'' and ``unitary'' simply by unitary. A standard
theorem of compact classical matrix groups is the following (See E.
Meckes' excellent notes \cite{EMeckes_IAS2014}). 
\begin{thm}
Let $G$ be any of $\mathbb{O}(N)$, $\mathbb{SO}(N)$, $\mathbb{U}(N)$
or $\mathbb{SU}(N)$. Then there is a unique translation-invariant
probability measure on G, which is called the Haar measure. 
\end{thm}
To explicitly and geometrically construct the Haar measure, one starts
by considering an $N\times N$ matrix $X$ whose entries are standard
gaussians. The joint probability density of the entries is:
\begin{equation}
d\mu_{\beta}(X)=\left(2\pi\right)^{-\beta N^{2}/2}\prod_{i,j=1}^{N}dx_{i,j}\text{ }\exp\left(-|x_{i,j}|^{2}/2\right),\label{eq:GaussianMeasure}
\end{equation}
where $\beta=1$ corresponds to real, $\beta=2$ to complex and $\beta=4$
to quaternion entries. It is easy to check that this measure is invariant
under a left-multiplication by a unitary matrix $V$. That
is, by the change of variables $y_{i,j}=[VX]_{i,j}$, one exactly obtains
the foregoing equations with $x_{i,j}$ replaced by $y_{i.j}$ everywhere.

For compact Lie groups, left translation invariance implies right translation invariance and therefore the Haar measure.

Although we have a measure that is invariant under left and right
multiplications by any unitary, the space of matrices being described
is not unitary. The standard way to proceed is to orthonormalize the
columns of $X$ by a QR-factorization $X=UR$, where $U$ is the desired
Haar measure unitary. $U$ is exactly the unitary matrix that results
by performing a Gram-Schmidt process on the columns of $X$. 
\begin{rem}
Although Gram-Schmidt results in unitaries from Haar measure, numerically stable QR algorithms such as Householder, do not perform Gram-Schmidt and can result in distributions different than Haar. To fix this, one must constrain the diagonal entries of $R$ to be positive and then use whatever blackbox QR algorithm available \cite{mezzadri2006generate}. 
\end{rem}
It is easy to see that the translation invariance of the space of such unitary matrices is preserved as the
Gram-Schmidt process (equivalently the QR-decomposition) commutes
with left multiplication by a fixed unitary matrix as we
now show.

To see this, let $X_{1},\dots,X_{N}$ be the columns of a matrix $X$.
The Gram-Schmidt process orthonormalizes these columns by systematically
removing the projections of the others. For example, starting from
$X_{1}$, the second column $X_{2}$ is replaced by $X_{2}-\langle X_{1},X_{2}\rangle X_{2}$.
Now suppose we left multiply the unitary matrix $V$, whereby the
columns become $VX_{1},VX_{2},\dots,VX_{N}$. The resulting second
column, by the unitarity of $V$, becomes
\[
VX_{2}-\langle VX_{1},VX_{2}\rangle VX_{1}=VX_{2}-\langle X_{1},X_{2}\rangle VX_{1}.
\]
The right hand side is equivalent to first doing Gram-Schmidt and
then multiplying the resulting vectors by $V$.  In other words, one can first apply a rotation to all the columns and then orthonormalize or first orthonormalize and then apply the rotation. 

In conclusion, if we start from the space of Gaussian matrices we have a unitary
invariant measure, which then through QR factorization results in
a translation invariant measure over the space of unitary matrices
(i.e., Haar measure).

We now show that we can achieve measures arbitrarily close the Haar
using pencils of matrices proposed in the interpolation in Eq. \eqref{eq:MatrixInterpolation}:
\begin{lem}
\label{lem:Theta_Haar}Consider the QR-decomposition $U(\theta)R(\theta)=(1-\theta)X+\theta\mathbb{I}$,
where the entries of $X$ are independently and identically distributed
gaussians $x_{i,j}\sim{\cal N}_{\beta}(0,1)$. Then for $\theta\ll1$,
the distribution over $U(\theta)$ is ${\cal O}(\theta)-$close to the Haar
measure.
\end{lem}
\begin{proof}
Let us denote by $Z(X;\theta)=(1-\theta)X+\theta\mathbb{I}$. Clearly,
$Z$ is a linear transformation of a gaussian random matrix $X$,
whose distribution is given in Eq. \eqref{eq:GaussianMeasure}. Therefore,
the distribution over $Z(X;\theta)$ denoted by $d\nu_{\beta}(Z(X))$
is:
\begin{eqnarray}
d\nu_{\beta}(Z(X)) & = & \left(2\pi(1-\theta)^{2}\right)^{-\beta N^{2}/2}\left(\prod_{i,j=1}^{N}dx_{i,j}\right)\exp\left[-\frac{1}{2(1-\theta)^{2}}\left(\sum_{i\ne j}|x_{i,j}|^{2}+\sum_{i}|x_{i,i}-\theta|^{2}\right)\right]\nonumber \\
 & = & d\mu_{\beta}(X)\left\{ 1+\theta\left[\beta N^{2}-\sum_{i}\left(|x_{i,i}|^{2}-\text{Re}(x_{i,i})\right)\right]+{\cal O}(\theta^{2})\right\}\label{eq:TVDformula},
\end{eqnarray}
where we assume $\theta\ll1$. Since the finite moments of a finite
gaussian matrix are bounded, the total variation distance between
$\mu_{\beta}(X)$ and $\nu_{\beta}(Z)$ obeys
\[
\left\Vert \mu_{\beta}(X)-\nu_{\beta}(Z(X))\right\Vert _{\text{TV}}\equiv\max_{A\subset\mathbb{R}}|\mu(A)-\nu(A)|\le O(\beta N^{2}\theta).
\]
To prove that $U$ and $U(\theta)$ that result from a QR decomposition have a TVD that is at most $O(\beta N^{2}\theta)$, we set up a coupling argument. Suppose $Y$ is distributed like $Z$ and is independent of $X$. We set up the coupling $(X,Y)$, whose marginals are just $X$ and $Y$. Applying the QR algorithm to the coupling results in the independent unitaries $(X',Y')$ and we need to argue that $TVD(X',Y')\le TVD(X,Y)\equiv \text{Pr}(X\ne Y)$. If we start with $X=Y$, the algorithm results in $X'=Y'$ by  the fact that $X$ and $Y$ are both with probability one invertible and the essential uniqueness of QR for invertible matrices. However, if $X\ne Y$ , the algorithm can result in $X' \ne Y'$ or $X' = Y'$. Therefore, $\text{Pr}(X'\ne Y')\le \text{Pr}(X\ne Y)=O(\beta N^{2}\theta)$. This proves that the TVD over the resulting unitaries is at most $O(\beta N^{2}\theta)$. 
\end{proof}
\begin{rem}
Intuitively, the output of any function of the random matrices $Z$
over the distribution $\nu_{\beta}(Z(X))$, maps to a distribution
whose distance is at most $O(\beta N^{2}\theta)$ from what would
have been obtained if the input were from the distribution $\mu_{\beta}(X)$.
Similarly, any algorithm whose input is from $\nu_{\beta}(Z)$, in
particular the distribution over $U(\theta)$ resulting from Gram-Schmidt
process, is also $O(\beta N^{2}\theta)$ close in distribution to
the space of unitary matrices from the Haar measure.
So we conclude
that the total variation distance between the distribution over $U(\theta)$ and the Haar
measure is $O(\beta N^{2}\theta)$.
\end{rem}
\section{\label{sec:RCS}Exact average-case $\#P$-hardness of random circuit
sampling (RCS)}
\subsection{Proof structure (informal)}
We now apply the foregoing mathematical results to a problem in quantum
computational complexity. It is known that there exist local circuits with $n$ qubits
whose probability amplitudes are $\#P$-Hard to estimate to within $1/\text{poly}(n)$
multiplicative error \cite{aaronson2011computational,bremner2011classical}. By a quantum circuit we have in mind a specific
architecture ${\cal A}$ that implements a specific unitary transformation
of $|0\rangle^{\otimes n}$. As stated in Subsection \eqref{sec:Summary_overview}
the unitary matrix $U$ that encodes the quantum computation is made
up of a product of unitary matrices with one unitary per layer of
the circuit. The unitary in each layer, in turn, is made up of a tensor
product of $1-$ or $2-$qubits unitaries (i.e., gates). Therefore,
if $C$ is the quantum circuit, it is a product of many gates $C=C_{m}C_{m-1}\dots C_{2}C_{1}$,
where each $C_{i}$ may be equal to a one qubit gate that implements
$\mathbb{I}\otimes U_{i}$ or a 2-qubit gate that implements $\mathbb{I}\otimes U_{ij}$.

\begin{defn}\label{RCS}
Below by random circuit sampling problem, or simply RCS, we mean sampling from the output distribution of a random circuit.
\end{defn}

The goal is to prove the following conjecture:
\begin{conjecture}\label{Conj1}
(Informal supremacy conjecture) Approximating to $1/\text{poly}(n)$
multiplicative error of most amplitudes of most quantum circuits is a \#P
hard problem.
\end{conjecture}
 
The statement we prove in this section informally reads: 
\begin{thm}
(Informal statement of the theorem) Exact calculation of most amplitudes
of most quantum circuits is a $\#P$-Hard problem.
\end{thm}
 
The key words that distinguish our results from the stated conjecture
are 'approximating' vs. 'exact calculation'. 

The proof structure is as follows. We start with an arbitrary circuit, which may be the worst case circuit
$C$, and whose exact specification is irrelevant for our purposes. We just assume
that it has an architecture ${\cal A}$ and as usual is made up of
local gates.

Our goal is to show that if the local gates were implementing local unitaries
independently drawn from the Haar measure, then the problem remains
 $\#P$-Hard to compute on average over this distribution. The way we do this is that we $\theta-$deform the worst
case circuit $C$ to be sufficiently close to such a random circuit
denoted by $C(\theta)$ such that $C(1)=C$. The proposed deformation
is rational function valued in $\theta$ for the quantity whose $\#P$-Hardness
we want to show. But in the previous section we provided an efficient
BW algorithm that can efficiently interpolate rational functions.
Therefore, we can invoke a Lipton-like argument (see italic text
in Section \eqref{sec:Summary_overview}) to prove that the average
case should also be $\#P$-Hard, because otherwise $C(1)$ would be easy
to compute by interpolation. 
\subsection{Formal results} 
By inserting a complete set of basis between each $C_{i}$ and $C_{i+1}$,
one can write the circuit down in what is at times called ``Feynman
path integral'' form \cite{bouland2018quantum}. The amplitude corresponding
to the initial state $|y_{0}\rangle$ and final state $|y_{m}\rangle$ is
\begin{equation}
\langle y_{m}|C|y_{0}\rangle=\sum_{y_{1},y_{2},\dots,y_{m-1}\in\{0,1\}^{n}}\langle y_{m}|C_{m}|y_{m-1}\rangle\langle y_{m-1}|C_{m-1}|y_{m-2}\rangle\cdots\langle y_{1}|C_{1}|y_{0}\rangle.\label{eq:Feynman}
\end{equation}

We now define some basic notation following \cite{bouland2018quantum}.
\begin{defn}
(Haar random circuit distribution) Let ${\cal A}$ be an architecture
over circuits and let ${\cal H}_{{\cal A}}$ be the distribution over
circuits in ${\cal A}$ whose gates are independently drawn from the
Haar measure. 
\end{defn}
The random circuit sampling is then:
\begin{defn}
(Random Circuit Sampling (RCS) \cite{bouland2018quantum}) Random
circuit sampling over a fixed architecture ${\cal A}$ is the the
following task: given a description of a random circuit $C$ from
${\cal H_{A}}$, and a description of an error parameter $\epsilon>0$,
sample from the probability distribution induced by $C$. That is
draw $y\in\{0,1\}^{n}$ with probability $\text{Pr}(y)=|\langle y|C|0\rangle|^{2}$
up to a total variation distance $\epsilon$ in time $\text{poly}(n,1/\epsilon)$.
\end{defn}
In RCS one seeks estimations of $|\langle y|C|0\rangle|^{2}$ but
any bit string $|y\rangle$ is simple to obtain by applying Pauli
$X$ matrices to positions in $|0^{n}\rangle$ that correspond to
1's. By this, so called 'hiding property' \cite{aaronson2011computational},
it is sufficient to prove the hardness of computing
\begin{equation}
\text{p}_{0}(C)\equiv|\langle0^{n}|C|0^{n}\rangle|^{2}.\label{eq:p0}
\end{equation}
\begin{conjecture}\label{QSupremacyConj}
(Quantum Supremacy Conjecture \cite{bouland2018quantum,aaronson2011computational})
There is no classical randomized algorithm that performs RCS to inverse polynomial total variation distance error. 
\end{conjecture}
\begin{rem}Conjecture \eqref{Conj1} implies Conjecture \eqref{QSupremacyConj} assuming Polynomial Hierarchy is infinite \cite{aaronson2011computational}.\end{rem}

In order to ultimately prove
this, some intermediate steps have been taken. It is known that the
estimation of the exact amplitudes of the worst case
circuit are $\#P$-Hard. More recently, Bouland et al proved that
the \textit{exact} average case RCS is also $\#P$-Hard \cite{bouland2018quantum}.
What remains open is to prove that the \textit{approximate average}
RCS is also $\#P$-Hard. This is what is meant by proving ``robustness''.

An application of the developments of the previous sections is a new
proof of the exact \#P-hardness of RCS, which among other things is
free of some issues presented in \cite{bouland2018quantum}. For
example, the truncation of the Taylor series is unsatisfactory from
a mathematical perspective. In particular, BW requires exact evaluation
of the polynomial at some of the many points, but the truncation introduces
errors at all points. One then has to do extra work to justify the
validity of the proof.
\begin{figure}
\begin{centering}
\includegraphics[scale=0.35]{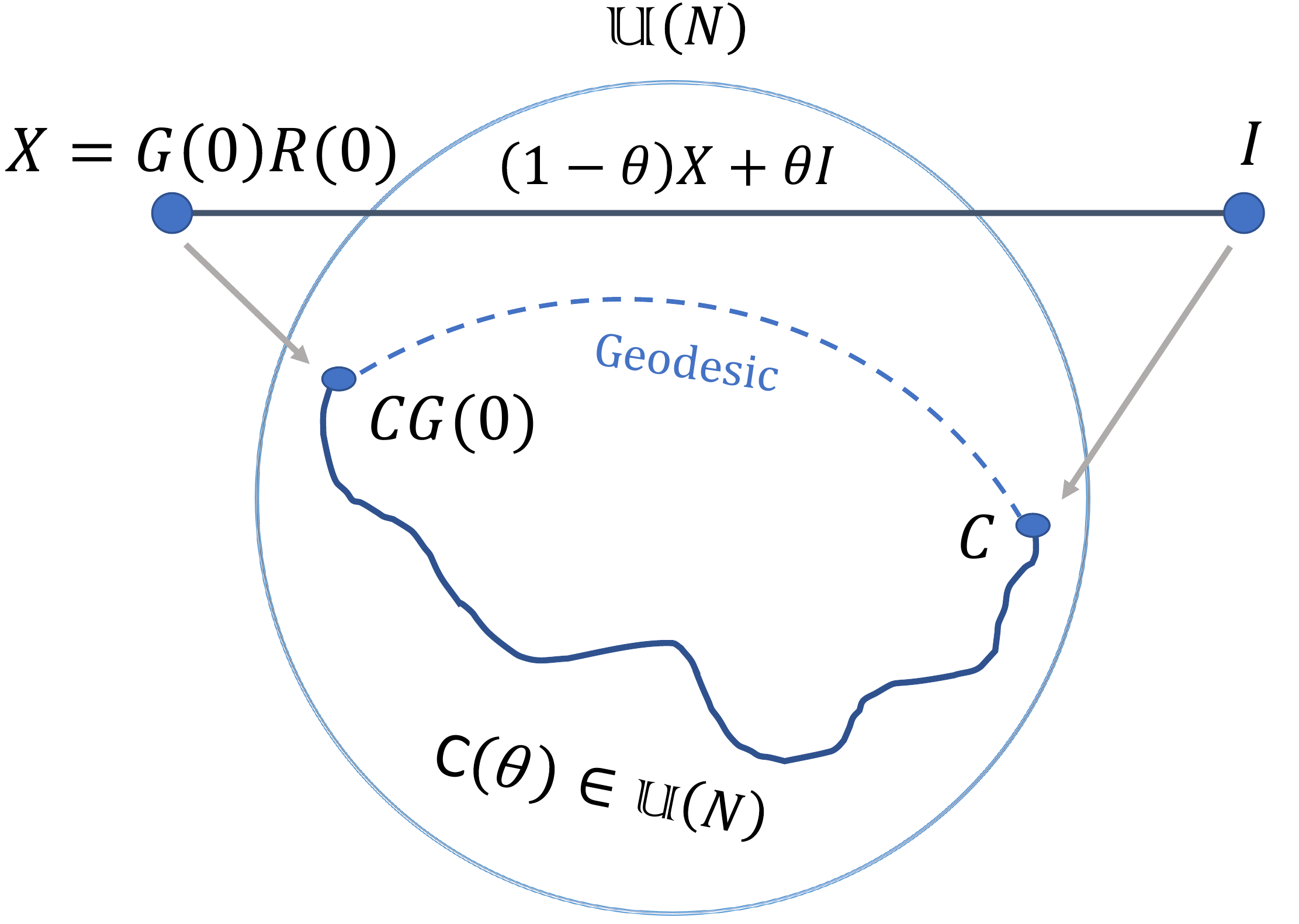}\caption{\label{fig:UnitaryG} Schematics of $C(\theta)\in\mathbb{U}(N)$ in Def. \eqref{def:C_theta}}
\par\end{centering}
\end{figure}
Therefore, we now proceed with the construction of the new reduction
of the exact-to-average case hardness based on constructing a path
on the unitary group described above.
\begin{defn}
\label{def:(-deformed-Haar)-Let}($\theta$-deformed Haar) Let ${\cal A}$
be an architecture over circuits, $\theta\in[0,1]$, and $G_{1},\dots,G_{m}$
be the gates in the architecture. Define the distribution ${\cal H}_{{\cal A},\theta}$
over circuits in ${\cal A}$ by setting each gate in the circuit to
be the unitary defined by the QR decomposition $G_{j}(\theta)R_{j}(\theta)=(1-\theta)X+\theta\mathbb{I}$,
where $X$ is a standard complex gaussian matrix with independent entries
$x_{ij}\sim{\cal N}_{\mathbb{C}}(0,1)$.
\end{defn}
Comment: $G_{j}(0)$ is a (small $2\times 2$ or $4 \times 4$) unitary distributed according to the Haar
measure (Sec. \eqref{sec:Close_to_Haar}) and $G_{j}(1)=\mathbb{I}$.
This interpolation is exact and no truncations are needed.
\begin{defn}
\label{def:C_theta}To reduce the complexity of the worst
case circuit $C=C_{m}C_{m-1}\dots C_{2}C_{1}$ to the average case, 
denote the latter by $C(\theta)=C_{m}(\theta)C_{m-1}(\theta)\dots C_{2}(\theta)C_{1}(\theta)$,
in which the $\theta-$deformed Haar gates are defined by $C_{j}(\theta)\equiv C_{j}G_{j}(\theta)$
and $G_{j}(\theta)$ is defined in Def. \eqref{def:(-deformed-Haar)-Let}.
Clearly (see Fig. \eqref{fig:Ctheta}),
\begin{eqnarray*}\label{Def:G}
\theta=0 & : & C_{j}G(0)\implies C(0)\in{\cal H}_{{\cal A}}\\
\theta=1 & : & C_{j}G(1)=C_{j}(1)=C_{j}\mathbb{I}\implies C(1)=C\quad\text{worst case circuit},
\end{eqnarray*}
the first relation follows from the translation invariance property
of the Haar measure that was discussed in Sec. \eqref{sec:Close_to_Haar} (see Fig. \eqref{fig:UnitaryG}).
Let us define the deformation of Eq. \eqref{eq:p0} by 
\begin{equation}
\text{p}_{0}(C(\theta))\equiv|\langle0^{n}|C(\theta)|0^{n}\rangle|^{2},\label{eq:p0_theta}
\end{equation}
which gives $\text{p}_{0}(C(1))=\text{p}_{0}(C)\equiv|\langle0^{n}|C|0^{n}\rangle|^{2}$. %
\textcolor{red}{}%
\end{defn}

\begin{figure}
\begin{centering}
\includegraphics[scale=0.5]{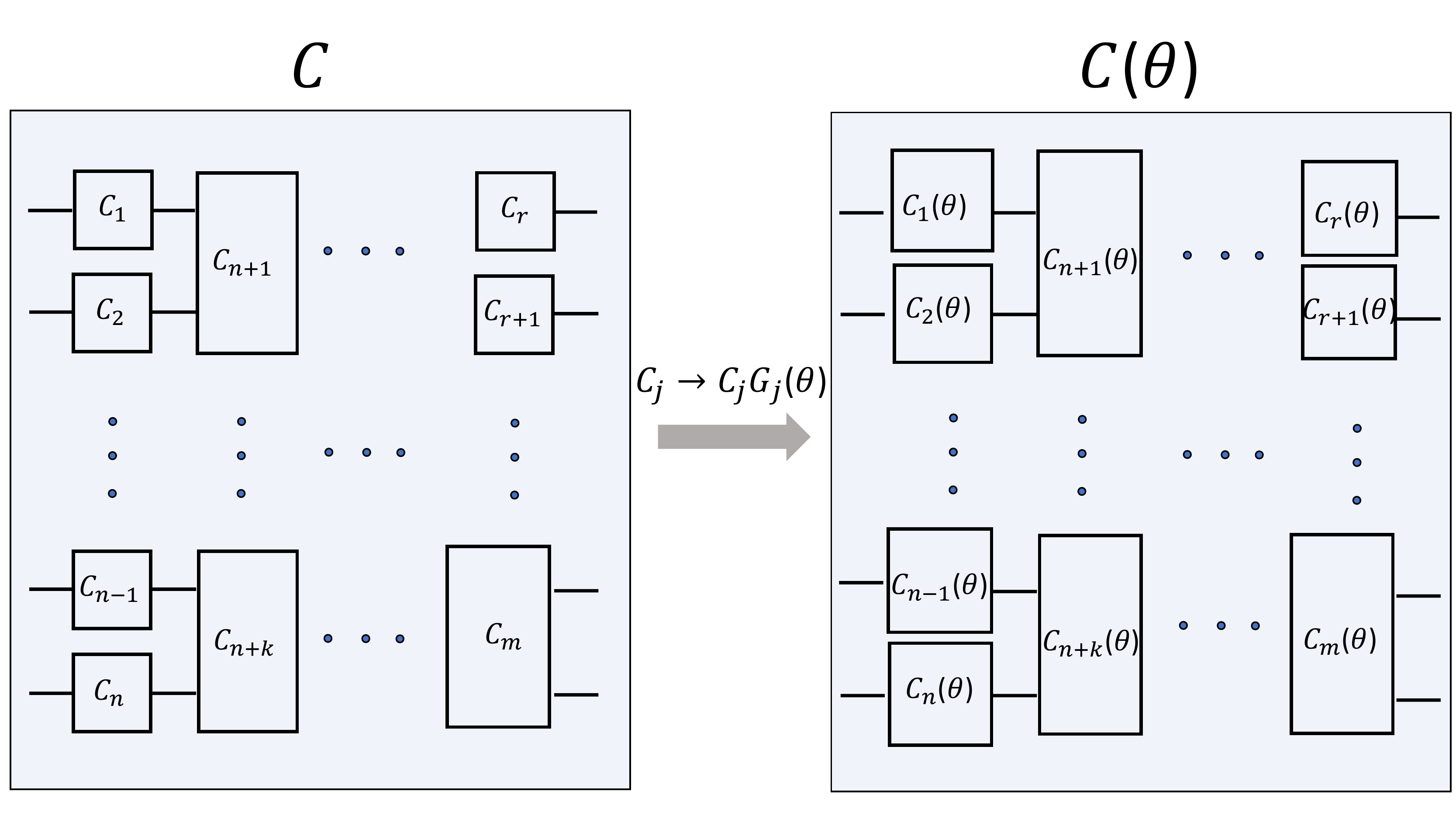}\caption{\label{fig:Ctheta}Schematics for Def. \eqref{def:C_theta}: The scrambling of the circuit $C$ to $C(\theta)$.}
\par\end{centering}
\end{figure}

We are now in a position to prove that the total variation distance
between circuits whose local gates are Haar random unitaries and $C(\theta)=C_{1}G_{1}(\theta)\text{ }C_{2}G_{2}(\theta)\cdots C_{m}G_{m}(\theta)$
is small. We emphasize that the scrambling takes place locally, which preserves the underlying architecture $\cal{A}$.
\begin{lem}
Let ${\cal A}$ be an architecture on circuits with $m$ gates. Then
$C.{\cal H}_{{\cal A}}$ is distributed according to ${\cal H_{A}}$,
and the total variation distance between $C\cdot{\cal H}_{{\cal A}}$
and $C\cdot{\cal H}_{{\cal A},\theta}$ is $O(\theta)$ for $\theta\ll1$.  
\end{lem}
\begin{proof}
Recall that Haar measure implies translational invariance with respect
to multiplications by fixed unitaries (see Subsection \eqref{sec:Close_to_Haar}).
Therefore, if $H_{j}$ is distributed according to the Haar measure
then so is $C_{j}H_{j}$. Moreover the $\ell_{1}$ norm that defines
total variation distance is invariant under unitary multiplication.
So it suffices to compare the measures over ${\cal H}_{{\cal A},\theta}$
and ${\cal H}_{{\cal A}}$, which by Lemma \eqref{lem:Theta_Haar} have
TVD of $O(\theta)$ over a single local gate. By the additivity of
TVD, the distribution induced by $C(\theta)$ which is denoted by
${\cal H}_{{\cal A},\theta}$ has a TVD from ${\cal H}_{{\cal A}}$
that is $O(m\theta)$. 
\end{proof}
 \begin{rem} 
We prove $O(\theta)$ TVD without requiring a priori that
$\theta=1/\text{poly}(n)$. In an $n-$ qubit circuit, $m=\text{poly}(n)$
and is often $m=O(n^{2})$. Therefore, if we take $\theta=1/\text{poly}(n)$
we are guaranteed for the TVD between $C\cdot{\cal H}_{{\cal A}}$
and $C(\theta)\in C\cdot{\cal H}_{{\cal A},\theta}$ to be $O(1/\text{poly}(n))$. 
\end{rem}
\begin{thm} 
\label{thm:MainResult}Let ${\cal A}$ be an architecture such that
$\text{p}_{0}(C)$ is $\#P-$Hard to calculate in the worst case.
Then it is $\#P$-hard to compute $3/4+1/\text{poly}(n)$ of the probabilities over ${\cal H}_{\cal A}$.\end{thm}
\begin{proof}
Consider $\text{p}_{0}(C(\theta))$ over the choice of $C(\theta)$ from
the distribution induced by the $\theta-$deformed Haar measure ${\cal H}_{{\cal A},\theta}$
over $G_{j}(\theta)$'s and for $\theta=1/\text{poly}(n)$.

Let $\{G_{j}\}$ be a collection of independent Haar random gates
as defined in Def. \eqref{def:(-deformed-Haar)-Let}. Recall from Def.
\eqref{def:C_theta} that $C(\theta)$ is the scrambled worst-case circuit,
which coincides with the worst case circuit at $\theta=1$. By the
modified QR-Algorithm each of the entries of $C(\theta)$ is a ratio
of a polynomial with the square root of a polynomial (see Eq. \eqref{eq:Entries_ModifiedQR}
in Alg. \eqref{alg:ModifiedQR}).

We now prove that $\text{p}_{0}(\theta)=\left|\langle0^{n}|C(\theta)|0^{n}\rangle\right|^{2}$
is a rational function of degree at most $(2mD_{4},2mD_{4})$, which
is low degree for $m=\text{poly}(n)$. This follows from Corollary \eqref{cor:Polynomial}
and the following arguments. Before normalizing each $\langle y_{j}|C_{j}(\theta)|y_{j-1}\rangle$
is a polynomial of degree at most $D_{4}$. The product of such
functions in Eq. \eqref{eq:Feynman} is simply a polynomial of degree at most $m D_4$ because there are $m$ terms in the product, and  these polynomials are closed under addition.
Therefore, $\langle0^{n}|C(\theta)|0^{n}\rangle$ is a polynomial of degree at
most $mD_{4}$ over the square root of a polynomial of degree at most
$2mD_{4}$. But we are interested in the probability $\text{p}_{0}(\theta)=|\langle0^{n}|C(\theta)|0^{n}\rangle|^2$,
which because of taking the absolute value and squaring becomes exactly a rational function of degree at most $(2mD_{4},2mD_{4})$. Recall that $D_4=27$ is just a constant.

Similar to the argument in \cite{bouland2018quantum}, a simple counting
shows that if there exists an oracle $\mathcal{O}$ that can compute
$\text{p}_{0}(\theta)$ for $3/4+1/\text{poly}(n)$ of $C(\theta)$,
then given that it correctly evaluates $\text{p}_{0}(\theta)$ over
$1/2+1/\text{poly}(n)$ of the choices of $\{U_{j}\}$, it must succeed
over $1/2+1/\text{poly}(n)$ of $\theta$'s as well. This sets the
tolerance to errors for interpolating $\text{p}_{0}(\theta)$ based
on sampling using the BW algorithm for rational functions (Alg. \eqref{alg:(Berlekamp-Welch-for-Rational}). 

We are seeking an interpolation of a $(2mD_{4},2mD_{4})$ rational
function. Since it is expected that $1/2-1/\text{poly}(n)$ of $\theta$'s
may be erroneous, the least expected number of independent evaluation
points ($\theta$'s) to exactly interpolate $F(\theta)$ is $2mD_{4}\text{poly}(n)$,
where the $\text{poly}(n)$ factor is exactly the one that appears
in the fraction of successful $\theta$'s. This exactly interpolates
$\text{p}_{0}(\theta)$ and we can now interpolate $\text{p}_{0}(1)$,
which is the worst case probability $\text{p}_{0}(C)$. Since the
latter is assumed to be a $\#P$-Hard problem, it would be a contradiction
to have an efficient oracle that evaluates $\text{p}_{0}(\theta)$
for generic $\theta$.
\end{proof}
\begin{rem}
Since the rational BW algorithm above is explicit and exact, it might
help characterize $\text{p}_{0}(\theta)=|\langle0^{n}|C(\theta)|0^{n}\rangle|^{2}$
sufficiently well as to help proving Conjecture \eqref{QSupremacyConj}.
\end{rem}
We now prove that Theorem \eqref{thm:MainResult} is necessary for the
quantum supremacy conjecture. This follows from closeness in total
variational distance of ${\cal H}_{{\cal A},\theta}$ to the Haar
measure local circuit ${\cal H}_{{\cal A}}$ that was essentially
in Lemma \eqref{lem:Theta_Haar} and Theorem \eqref{thm:MainResult}.
Namely, we show that any algorithm ${\cal O}$ that works on average over
circuits drawn from the distribution ${\cal H}_{{\cal A},\theta}$
can be used to get an algorithm that works on average
over circuits drawn from ${\cal H}_{{\cal A}}$.
\begin{cor} Exact average case hardness (Theorem \eqref{thm:MainResult}) is necessary for proving the quantum supremacy conjecture (conjecture \eqref{QSupremacyConj}).
\end{cor}
\begin{proof}
We prove that conjecture  \eqref{QSupremacyConj} implies Theorem \eqref{thm:MainResult}.  Suppose there exists
an algorithm ${\cal O}$ whose input is a circuit $C(0)$ with local
gates being Haar distributed (i.e., ${\cal H}_{{\cal A}}$ distributed)
and outputs $\text{p}_{0}(C(0))$. Then the same algorithm run on
the circuit $C(\theta)$ whose local gates are distributed according
to ${\cal H}_{{\cal A},\theta}$ in Lemma \eqref{lem:Theta_Haar} results
in the output $\text{p}_{0}(C(\theta))$, whose distribution is $O(\theta)$
close to the distribution of $\text{p}_{0}(C(0))$.  This is because the total variation distance is the supremum
over the positive difference of all functions between the two distributions.
Since the algorithm ${\cal O}$ is such an example, on inputs ${\cal H}_{{\cal A}}$
and ${\cal H}_{{\cal A},\theta}$, it results in output distributions
that are at most $O(\theta)$ in total variation distance. Now if there exists an algorithm that exactly computes probabilities on average (with probability $3/4+1/\text{poly}(n)$), then the algorithm approximately computes the probabilities over ${\cal H}_{{\cal A}}$. Hence if the latter is $\#P$-hard then the former must be too.
\end{proof}

\begin{rem}
We proved the exact average case hardness of sampling from the output distribution of RCS. The previous work of Bouland et al proves a small robustness of $2^{-\text{poly}(n)}$ for the truncated polynomial series expansion \cite{bouland2018quantum}. Since the denominators of the rational functions being considered here can be proved to be not too small, we expect similar arguments based on Patrui's lemma \cite{paturi1992degree} would prove exponentially small robustness in our case as well. However, we leave a more thorough characterization of the robustness based on the ideas herein for future work.
\end{rem}

\section{Discussion and conclusions}
Here we constructed a general way of interpolating between unitaries (e.g., quantum gates or circuits), which may find independent use in other contexts such as those involving scrambling. The interpolation is efficient in that certain functions of it can be learned by a polynomial number of samples. This follows from the polynomial and rational function dependence of the entries on the interpolation parameter (Alg. \eqref{alg:ModifiedQR}). We also showed that the unitary matrices that result may be controlled to have a distribution close to Haar measure in total variation distance. 

The mathematical results were then applied to give a fresh new proof of the average case hardness of random circuit sampling. There are other known facts about RCS, such as the equivalence with sampling and anti-concentration results.  We refer the reader to Appendices A5-A7 of \cite{bouland2018quantum} for a summary and other complexity theoretical background on this problem. 

It is important to note that in the definition of RCS, one wants to
allow some errors, partly because the errors are experimentally inevitable.
The complexity of RCS is an open problem. To prove the quantum supremacy conjecture the robustness needs to be improved. However, the quantum supremacy conjecture may be false or unrealistic to realize experimentally. 

\begin{center}
\bf{Acknowledgements}
\end{center}
I thank Adam Bouland, Bill Fefferman, Karthikeyan Shanmugam, Scott Aaronson, Madhu Sudan
for discussions. I acknowledge the support of the Frontiers Foundation
and the MIT-IBM collaborative grant.
\end{singlespace}
\bibliographystyle{plain}
\bibliography{mybib}

\end{document}